\documentclass[11pt]{article}

\newcommand{\version}{March 7, 2014}

\usepackage{amsthm,amsfonts,amsmath, amscd}
\usepackage{amssymb,amsmath, dsfont}
\usepackage{graphicx}

\swapnumbers
\theoremstyle{plain}

\newtheorem{thm}{THEOREM}[section]

\theoremstyle{definition}

\theoremstyle{definition}

\newcommand{\upchi}{\raise1pt\hbox{$\chi$}}

\newcommand{\tr}{{\rm Tr}}

\numberwithin{equation}{section}
\def\H{\mathcal{H}}

\begin{document}


\def\tr{{\rm Tr}}

\title{Remainder Terms for Some Quantum Entropy Inequalities}
\author{\vspace{5pt} Eric A. Carlen  and
Elliott H. Lieb \\
\vspace{5pt}\small{$1.$ Department of Mathematics, Hill Center,}\\[-6pt]
\small{Rutgers University,
110 Frelinghuysen Road
Piscataway NJ 08854-8019 USA}\\
\vspace{5pt}\small{$2.$ Departments of Mathematics and
Physics, Jadwin
Hall,} \\[-6pt]
\small{Princeton University, Washington Road, Princeton, NJ
  08544-0001}\\
 }
\date{\version}
\maketitle

\let\thefootnote\relax\footnote{
\copyright \, 2014 by the authors. This paper may be
reproduced, in its
entirety, for non-commercial purposes.}

\begin{abstract}
We consider three von Neumann entropy inequalities: subadditivity;
Pinsker's inequality for relative entropy; and the monotonicity  of 
relative entropy. For these we state conditions for equality,  and we prove 
some new error bounds away from
equality, including an improved version of Pinsker's inequality.
\end{abstract}

\medskip
\leftline{\footnotesize{\qquad Mathematics subject
classification numbers: 81V99, 82B10, 94A17}}
\leftline{\footnotesize{\qquad Key Words: Density matrix, Entropy, Partial
trace }}

\section{Introduction}

We are concerned with the von Neumann entropy of a density matrix $\rho$
defined by $S(\rho) = -\tr \rho \ln \rho$. When $\rho_{12}$
is defined on $\H_1\otimes \H_2$, let $\rho_1 = \tr_2\rho_{12}$, etc., 
Let $S_{12}$ denote the entropy of $\rho_{12}$, $S_1$ the entropy of
$\rho_1$ and $S_2$ the entropy of $\rho_2$. 

Subadditivity of entropy \cite{l} states that $S_{12} \leq S_1 +S_2$ and
that
there is
equality if and only if $\rho_{12}  = \rho_1 \otimes \rho_2$. 
The kind of inequality we consider here is one that gives a (sharp)  
remainder term whenever $\rho_{12}$ is not of this product form. 
Inequalities of this type are known, as we explain below,  but, unlike the
one we prove here, 
they are only saturated when $S_{12} =S_1 +S_2 $.

The method by which we find and prove such a remainder term turns out to
be useful for other entropy inequalities, among which we discuss  an
improved
remainder term for the positivity of relative entropy (Pinsker-Cziszar)
and the
monotonicity
of relative entropy.

Note that proving the existence of a remainder term requires, by
definition, knowing the cases of equality. Thus, our results necessarily
include statements about the minimizers of the various entropy functionals. 
A particular example is a new condition for equality in the monotonicity of
relative entropy, which is not obviously the same as the known condition.
We shall only speak of von Neumann (and Shannon, in the
classical cases)  entropies here.

\section{Theorems} \label{intro}

\subsection{Quantitative subadditivity (mutual information)}

The mutual information of two subsystems in a combined state $\rho_{12}$ is
defined to be $S_{1} +S_2 -S_{12}$ which is non-negative by the
subadditivity of entropy. It is known that this is zero if and only it
$\rho_{12}$ is a product state. The following quantifies mutual information
in terms of departure from the product state condition.
\begin{thm}[Quantitative subadditivity]\label{quant111}
\begin{equation}\label{111state}
S_1 +S_2 - S_{12} \geq  -2\ln\left(1 - \tfrac12 \tr
\left[\sqrt{\rho_{12}}-\sqrt{\rho_1\otimes \rho_2}\right]^2\right).
\end{equation}
In particular, $S_1 +S_2 - S_{12} \geq 0$ with equality if and only if 
$\rho_{12} = \rho_1\otimes \rho_2$. 

More generally, with an obvious notation, if 
$\rho_{1\cdots N}$ is a density matrix on $\H_1\otimes \cdots \otimes
\H_N$
then  

\begin{equation} \label{nstate}
\sum_1^N S_j- S_{1\cdots N}
\geq -2\ln\left(1 - \tfrac12\tr 
\left[ \sqrt{\rho_{1\cdots N} } - \sqrt{\rho_1 \otimes \cdots \otimes
\rho_N}\right]^2\right).
\end{equation}
\end{thm}

\begin{proof}
Recall the Peierls-Bogoliubov inequality: If $H$ and $A$ are self
adjoint
operators and $\tr e^{-H} =1$, then 
\begin{equation}\label{PB}
 \tr\left(e^{-H+A}\right) \geq \exp(\tr A\, e^{-H})\ .
\end{equation}
To prove (\ref{111state}), apply this with 
$$H = - \ln \rho_{12} \qquad{\rm and}\qquad A = \tfrac12(\ln \rho_1 +
\ln\rho_2 - \ln\rho_{12})\ .$$
Then with $\Delta:= \tfrac12 (S_{12} - S_1 - S_2)$, by the
Peierls-Bogoliubov
inequality and the Golden-Thompson inequality,
\begin{eqnarray}
e^\Delta &=& \exp \left [\tr\rho_{12}\tfrac12 (\ln \rho_1 + \ln\rho_2
-
\ln\rho_{12})\right]\nonumber\\
&\leq &  \tr \exp\left[ \tfrac12 (\ln \rho_{12} +  \ln(\rho_1\otimes
\rho_2)\right]\nonumber\\
&\leq &  \tr \exp\left[ \tfrac12 \ln \rho_{12}
\right]\exp\left[\tfrac12 
\ln(\rho_1\otimes \rho_2)\right]\nonumber\\
&=& \tr \left[ \rho_{12}^{1/2} (\rho_1\otimes \rho_2)^{1/2} \right]\
.\nonumber
\end{eqnarray}
Since
\begin{equation}\label{proj4}
\tr \left[ \rho_{12}^{1/2} (\rho_1\otimes \rho_2)^{1/2} \right]  =
\left(1 - 
\tfrac12\tr \left[ \rho_{12}^{1/2} -(\rho_1\otimes \rho_2)^{1/2}
\right]^2\right)\ ,
\end{equation}
this proves (\ref{111state}). An obvious adaptation proves
(\ref{nstate}).  
\end{proof}

Another way to obtain a quantitative subadditivity bound is to use the fact
that subadditivity can be seen as a consequence 
of the positivity of relative entropy, and that the well known Pinsker inequlity provides a 
positive lower bound on the relative entropy. 

Let $\rho$ and $\sigma$  be density matrices. The {\em relative entropy} of
$\rho$ with respect to $\sigma$,
$D(\rho||\sigma)$, is defined by
$$D(\rho||\sigma)  =  \tr \left(\rho\, [ \ln \rho - \ln \sigma]\right)\ .$$
One readily computes that 
\begin{equation}\label{subrel}
S_1+S_2 - S_{12} = D(\rho_{12}|| \rho_1\otimes \rho_2)\ .
\end{equation}
Pinsker's inequality \cite{Pin,Csi} states that 
\begin{equation}\label{pinsk}
D(\rho||\sigma)  \geq    
\frac12 \left(\tr | \rho -\sigma| \right)^2\ .
\end{equation}
There is equality in \eqref{pinsk} only when both sides are zero. Indeed,
refinements of Pinsker's inequality can be found in \cite{FHT} in the form
of an expansion in powers of $\left(\tr | \rho -\sigma| \right)^2$   with
positive coefficients, 
the right side of
\eqref{pinsk} being the first.

The combination of (\ref{subrel}) and \eqref{pinsk}   yields
\begin{equation}\label{pinskA}
S_1+S_2 - S_{12} \geq \frac12  \left(\tr | \rho_{12} -\rho_1\otimes \rho_2|\right)^2\ .
\end{equation}
By what we just explained above about Pinsker's inequality there is
equality in \eqref{pinskA} only when the mutual information is zero. 
The bound provided by Theorem \ref{quant111}, however, is sharp in
situations in which mutual information is not small, as the following
example shows. 

\medskip
{\it Example:}
Let $\H$ be an $N$-dimensional Hilbert space,
and consider  the (unique)
 $N$-particle Slater determinantal state $\H ^{\wedge N}$. Let $\rho$ be its
two-particle reduced density
 matrix, which is the normalized projection onto $\H \wedge \H \in \H
\otimes \H$.  Let $\sigma$ be the 
 tensor product of the one particle reduced density
 matrix with itself, so that $\sigma$ is the normalized identity. One then
readily computes two equalities:
\begin{equation}
 D(\rho||\sigma) =  
\ln\left(\frac{2N}{N-1}\right) = -2\ln  \tr \left[ \rho^{1/2} \sigma^{1/2}
\right].
\end{equation}
and this is exceeds $\ln(2) \approx 0.693$. 
Together with \eqref{subrel}, this shows that the two sides of
inequality \eqref{111state} are equal. 

On the other hand, the 
right side of \eqref{pinskA} is 
$$ \tfrac12(\tr |\rho - \sigma|)^2 =  \tfrac12
\left(\frac{N+1}{N}\right)^2 \ .
$$
which is approximately $0.5$ for large $N$. 

 Theorem \ref{quant111} is used in \cite{CLF} to prove extremal properties
of Slater determinantal states for entropy and measures of entanglement.
This proof
requires precisely the kind of sharpness displayed in this example.

\medskip

{\it Second proof of Theorem \ref{quant111}:}
Instead of using Pinsker's inequality to bound $ D(\rho||\sigma)  $ from
below, one can use the Renyi entropy bound
 \begin{equation}\label{222}
 D(\rho||\sigma)  \geq   -2\ln  \tr \left[ \rho^{1/2} \sigma^{1/2} \right]
\ .
 \end{equation}
Our second proof of Theorem \ref{quant111} is obtained by combining
\eqref{222} with \eqref{subrel}.

M. Wilde has pointed out to us that \eqref{222} is a consequence of
the monotonicity of the relative Renyi entropy $D_\alpha(\rho|| \sigma)
= \tfrac{1}{\alpha-1} \ln \left(   \tr \rho^\alpha \sigma^{1-\alpha}
\right)$
with
respect to the 
parameter $\alpha \in (0,1)$. One compares $D_\alpha $ at $\alpha =
\tfrac12$ and 
$\alpha \to 1$. 

Alternatively, the method of proof of Theorem \ref{quant111} gives another
simple proof of \eqref{222}, as follows.
Apply (\ref{PB}) with 
$H = - \ln \rho $ and $A = \tfrac12(\ln \sigma  - \ln \rho)\
.$
Then with $\Delta:= \frac12 D(\rho||\sigma)$, by the Peierls-Bogoliubov
inequality and the Golden-Thompson inequality,
\begin{eqnarray}
e^{-\Delta} &=& \exp \left [\tr\rho \frac12 (\ln \sigma  -
\ln\rho)\right]
\leq   \tr \exp\left[ \frac12 (\ln\sigma  +  \ln \rho)\right]\nonumber\\
&\leq &  \tr \exp\left[ \frac12 \ln \sigma     \right]\exp\left[\frac12 
\ln\rho\right]
= \tr \left[ \sigma^{1/2} \rho^{1/2} \right]\ .
\end{eqnarray}

\subsection{Monotonicity of Relative Entropy}

Another application of this method, this time using a deeper theorem -- the
{\it triple-matrix generalization} of the Golden-Thompson inequality 
\cite{wy}. 
Monotonicity of relative entropy is the inequality
\begin{equation} \label{difrel}
D(\rho_{12} ||\sigma_{12}) \geq D(\rho_{1} ||\sigma_{1}).
\end{equation}
and our goal is to find a remainder term for the difference of these two 
quantities.  

A more general result follows immediately. Stinespring's Theorem  says 
any CPT (completely positive trace preserving) 
map $T$ can writen as a partial trace composed with
a unitary transformation, and thus (\ref{difrel}) implies
\begin{equation} \label{difrel2}
D(\rho ||\sigma) \geq D(T(\rho) ||T(\sigma)) 
\end{equation}
for all density matrices $\rho,\sigma$ on $\H$ and all CPT maps
$T:\mathcal{B}(\H) \to \mathcal{B}(\mathcal{K})$, where
$\mathcal{K}$ is another Hilbert space. 

\begin{thm}[Relative entropy monotonicity bound]
\begin{equation} \label{eqln1}
D(\rho_{12} ||\sigma_{12}) - D(\rho_{1} ||\sigma_{1}) \geq
\tr \left[ \sqrt \rho_{12} - 
\exp\left\{ \tfrac12 
\ln\sigma_{12} - \tfrac12 \ln \sigma_1 +\tfrac12  \ln \rho_1 
\right\}\right]^2.
\end{equation}
In particular,
there is equality in \eqref{difrel} if and only if 
\begin{equation} \label{eqln}
 \ln \rho_{12} - \ln \sigma_{12} =  \ln \rho_{1}\otimes I_2 - \ln
\sigma_{1} \otimes I_2.
\end{equation}
\end{thm}

The equality condition (\ref{eqln}) was given by Ruskai in \cite{Rus}. The
remainder term (\ref{eqln1}) is new. 

\begin{proof}
We introduce 
\begin{equation}
\Delta = \tfrac12 D(\rho_{12} ||\sigma_{12}) - \tfrac12 D(\rho_{1}
||\sigma_{1}) =   \tfrac12 \tr_{12}\  \rho_{12} \left\{\ln \rho_{12} -
\ln \sigma_{12} +\ln \sigma_1 - \ln \rho_1 \right\},
\end{equation}
where we omit the $\otimes I_2$ for simplicity of notation.  

We apply the Peierls- Bogoliubov inequality as before:
$$
{\rm e}^{- \Delta} \leq \tr \exp \tfrac12 \{ \ln \rho_{12} +\ln
\sigma_{12} - \ln \sigma_1 + \ln \rho_1  \}.
$$
We then employ the Golden-Thompson and the Schwarz inequalities to conclude
that 
\begin{multline}
{\rm e}^{-\Delta} \leq \tr \sqrt{\rho_{12} }    
\exp\left\{ \tfrac12 
\ln \sigma_{12} - \tfrac12 \ln \sigma_1 +\tfrac12  \ln \rho_1  \right\}\\
\leq \left(\tr \rho_{12} \right)^{1/2}   \left( \tr  \exp\left\{ 
\ln\sigma_{12} -  \ln \sigma_1 +\ln \rho_1  \right\} \right)^{1/2}
\nonumber\end{multline}
By the triple matrix generalization of the Golden-Thompson inequality
\cite{wy}
$$
\tr \exp\left\{ 
\ln \sigma_{12} -  \ln \sigma_1 +\ln \rho_1  \right\} \leq  
\tr \int_0^\infty \sigma_{12} \   (t + \sigma_1)^{-1} \rho_1 \ (t +
\sigma_1)^{-1}  {\rm d}t
= \tr_1 \ \rho_1  = 1.
$$
Therefore, 
\begin{multline}
{\rm e}^{-\Delta} \leq 
\tr \sqrt{\rho_{12} }    
\exp\left\{ \tfrac12
\ln\sigma_{12} - \tfrac12 \ln \sigma_1 +\tfrac12  \ln \rho_1  \right\}
\leq \\ 
1- \tfrac12 \tr \left[ \sqrt \rho_{12} - 
\exp\left\{ \tfrac12 
\ln\sigma_{12} - \tfrac12 \ln \sigma_1 +\tfrac12  \ln \rho_1 
\right\}\right]^2.
\end{multline}
Since $ -\ln(1-y) \geq y$, as stated before,
$$
\Delta \geq   \tfrac12 \tr \left[ \sqrt \rho_{12} - 
\exp\left\{ \tfrac12 
\ln\sigma_{12} - \tfrac12 \ln \sigma_1 +\tfrac12  \ln \rho_1 
\right\}\right]^2,
$$
and this is zero if and only if \eqref{eqln} is satisfied.
\end{proof}

Petz \cite{P} has shown that $\Delta =0 $ if and only if $
T_\rho (\sigma_1) =\sigma_{12}$, where, for any density matrix $\tau_1$ on
$\H_1$, 
$$
T_\rho(\tau_1)  = \rho_{12}^{1/2} \, \rho_1^{-1/2} \, \tau_1 \, 
\rho_1^{-1/2} \rho_{12}^{1/2}.
$$
This $T_\rho$ is a CPT map.    As a consequence of the monotonicity of the 
relative entropy, we have $D(  T_\rho (\rho_1)) || T_\rho (\sigma_1) \leq
D(\rho_1 || \sigma_1)$. 
Since $T_\rho (\rho_1) =\rho_{12}$, by construction, Petz's condition
$T_\rho (\sigma_1) =\sigma_{12}$ gives 
$$D(\rho_{12}||\sigma_{12}) = D(  T_\rho (\rho_1)) || T_\rho (\sigma_1))\
.$$
Combining this with (\ref{difrel}) and (\ref{difrel2}), we obtain
$$D(\rho_{12}||\sigma_{12})  \geq D(\rho_1||\sigma_1) \geq  D(  T_\rho
(\rho_1)) 
|| T_\rho (\sigma_1) = D(\rho_{12}||\sigma_{12}) \ ,$$
and hence all inequalities must be equalities. 
We
see, therefore,  that Petz's condition is a sufficient condition for
equality. 
In contrast, condition (\ref{eqln}) is immediately seen to be
sufficient by direct calculation.
Petz's  condition and (\ref{eqln}) are equivalent, but it is not easy to
see this. 

Neither condition (\ref{eqln}), nor Petz's equivalent condition $\sigma_{12}
= T_\rho (\sigma_1)$, is transparent. It was  Hayden, Jozsa, Petz and
Winter \cite{HJPW} who characterized the solution set for these equations
to hold. 
They are satisfied if and only if
\begin{equation}
\H_1\otimes \H_2 = \oplus_{j=1}^N \H_{1,j}\otimes \H_{2,j} 
\end{equation}
\begin{equation}
 \rho_{12} = \oplus_{j=1}^N q_j \omega_j\otimes \rho_{2,j}\quad{\rm
and}\quad 
\sigma_{12} = \oplus_{j=1}^N r_j \omega_j\otimes \sigma_{2,j}
\end{equation}
where the $q_j,r_j$ are non-negative numbers summing to one and for each
$j$,
$\rho_{2,j},\sigma_{2,j}$ are density matrices on $\H_{2,j}$, and
$\omega_j$ is a density matrix on $\H_{1,j}$.

Monotonicity of the relative entropy provides one route to strong
subadditivity (SSA) of quantum entropy, as in \cite{LR}. One of the
motivations for seeking remainder terms in the monotonicity of relative
entropy is to obtain remainder terms for SSA. One such remainder term was
given in our paper \cite{cl}.

{\bf Acknowlegements:} We thank Mark Wilde and Lin Zhang for a careful
reading of an
earlier version of this paper and for useful comments.  This
work was
partially
supported by U.S. National Science Foundation
grants DMS-1201354 (E.A.C.) and PHY-0965859  and PHY-1265118 (E.H.L.).

\end{document}